\documentclass[letterpaper, 10 pt, conference]{ieeeconf}

\usepackage[hypertexnames=false,hidelinks]{hyperref}
\usepackage[all]{hypcap}
\usepackage[noadjust]{cite}
\usepackage[numbered,depth=subsubsection]{bookmark}
\usepackage{silence,etoolbox,graphicx,float,dsfont,color}
\usepackage{regexpatch,microtype,pgffor,bm}
\usepackage{amssymb,mathtools,commath,amsmath,algorithm,algpseudocode}
\usepackage{array,booktabs,subfigure}
\usepackage[capitalize]{cleveref}

\newcommand\makemathcal[1]{\foreach\z in{#1}{\expandafter\gdef\csname\z\expandafter\endcsname\expandafter{\expandafter\mathcal\expandafter{\z}}}}

\makemathcal{A,...,Z}

\newcommand{\RR}{\mathbb R}

\DeclareMathOperator*{\argmax}{arg\,max}

\let\oldforall\forall
\let\forall\undefined
\DeclareMathOperator{\forall}{\oldforall}

\BeforeBeginEnvironment{tabular}{\small\centering}
\graphicspath{{figures/}}

\crefname{equation}{}{}

\newtheorem{theorem}{Theorem}
\newtheorem{problem}{Problem}

\newtheorem{definition}{Definition}

\title{\LARGE \bf A Fast Algorithm for Robust Action Selection in Multi-Agent Systems}
\author{Jun Liu and Ryan K. Williams%
\thanks{This work was supported by the National Institute of Food and Agriculture under Grant 2018-67007-28380.}
\thanks{The authors are with the Department of Electrical and Computer Engineering, Virginia Polytechnic Institute and State University, Blacksburg, VA 24061 USA
(e-mail: \href{mailto:junliu@vt.edu}{junliu@vt.edu}; \href{mailto:rywilli1@vt.edu}{rywilli1@vt.edu}).}}

\IEEEoverridecommandlockouts
\begin{document}

\maketitle
\thispagestyle{empty}


\begin{abstract}
	In this paper, we consider a robust action selection problem in multi-agent systems where performance must be guaranteed when the system suffers a worst-case attack on its agents. Specifically, agents are tasked with selecting actions from a common ground set according to \emph{individualized} objective functions, and we aim to protect the system against attacks. In our problem formulation, attackers attempt to disrupt the system by removing an agent's contribution after knowing the system solution and thus can attack perfectly. To protect the multi-agent system against such attacks, we aim to maximize the minimum performance of all agents' individual objective functions under attacks. Thus, we propose a \emph{fast algorithm} with tunable parameters for balancing complexity and performance, yielding substantially improved time complexity and performance compared to recent methods. Finally, we provide Monte Carlo simulations to demonstrate the performance of the proposed algorithm.
\end{abstract}


\section{Introduction}

In multi-agent systems, action selection problems have been used in various applications, including task assignment \cite{williams2017decentralized}, path planning \cite{singh2009efficient,heintzman2021anticipatory}, coverage \cite{liu2021distributed}, environmental monitoring \cite{liu2020monitoring}, rigidity evaluation \cite{williams2014evaluating}, sensor selection \cite{liu2018optimal}, etc. In general, problems of this type are difficult combinatorial optimization problems, and recent work in this area has focused extensively on (non-)submodular optimization algorithms, which have proven tremendously useful as they offer performance guarantees coupled with efficient greedy selection.

In the general setting of submodular optimization, this paper focuses on a scenario where a team of agents can select actions from a common ground set to fulfill a system-level objective. By using multiple agents for a single collaborative task, the system's redundancy can be increased to deal with possible attacks or failures. For example, in a collaborative task assignment problem, a team of agents (e.g., robots) may need to move an object from one location to another by applying different skills. If one agent is attacked, and the attack is a worst-case (the best possible attack), we desire a system that can continue operation with minimal impact on performance. By worst-case attack, we refer to the case when the attacker knows the best selection of the system and tries to remove the contribution from the agent with the maximum contribution. Here, the contribution is quantified by each agent's individual objective function value. To protect against such attacks, we must ensure that the minimum contribution from any agent is maximized. In such a case, when the above attack occurs we can guarantee our system's performance. This paper will tackle the problem by constructing an appropriate submodular function from the system objective function and providing a fast and tunable algorithm for finding \emph{robust} action selections.

To provide context for our work and argue novelty, we now provide a review of relevant efforts in submodular optimization. Generally, submodular functions are a subset of set functions with a diminishing returns-like property. This property can be used to model many reward functions as the marginal reward of a system often decreases as the system makes more selections. For example, in the well-known sensor placement problem \cite{krause2014submodular}, the authors utilized the submodularity of the objective function, i.e., mutual information, and proposed a greedy method to solve the problem with guarantees. In the multi-agent task allocation problem \cite{williams2017decentralized}, the marginal reward of the system is generally reduced when more tasks are assigned. Similar applications in robotics can also be found in the coverage problem \cite{liu2021distributed}, orienteering problem \cite{jorgensen2017matroid}, interaction planning \cite{heintzman2020multi}, target tracking \cite{zhou2018resilient}, precision agriculture monitoring \cite{liu2021intermittent}, etc.

While submodular functions (and thus submodular optimization) have broad applications, it is also essential to ensure a system's optimization performance when attacked, removing some parts or the entire contribution from one or more agents. Such scenarios have attracted significant attention as of late, and several robust/resilient algorithms have been proposed to mitigate the impact of attacks. In \cite{tzoumas2018resilient}, the authors considered a general action selection (or task allocation) scenario. They proposed an algorithm to protect the system when an attacker knows the system's solution and can remove part of the system's solution under a partition matroid constraint. In \cite{liu2021distributed}, we proposed a distributed version of that algorithm that can be used in situations where no central point of command is available. The proposed algorithm has the same performance guarantees as to the centralized algorithm. In \cite{krause2007selecting}, the author considered the case where a multi-agent team must select actions from a common ground set but with different individual objective functions. This is also the scenario that we consider in this paper. To protect the system against attacks, the authors proposed an algorithm by utilizing the submodularity of a surrogate function instead of the original objective function. The proposed algorithm works for the cardinality constraint scenario, and the final performance is guaranteed when the constraint is relaxed. However, under some budget-critical circumstances, this relaxation cannot be allowed. In \cite{hou2021robust}, the authors investigated this problem under a special case where the marginal gains of different agents' objective functions have a relation. The proposed algorithm's performance is guaranteed and is a function of the proposed marginal gain ratio. In contrast, we investigate this problem without any additional assumptions and aim to increase the speed of calculating a solution, which is useful when it is time-consuming to evaluate objective functions, a typical case.

In summary, the contributions of this paper are as follows.
\begin{itemize}
	\item We propose a fast algorithm for solving the robust multi-agent action selection problem.
	\item We prove the proposed algorithm's performance and computational complexity.
	\item We demonstrate the performance of the proposed algorithm using an action selection application.
\end{itemize}


\section{Preliminaries and Problem Formulation}
\label{sec: preliminaries and problem formulation}

\subsection{Preliminaries}
\label{ssec: preliminaries}

As our problem will be formulated and solved with tools from combinatorial optimization, we begin by reviewing relevant concepts in that area.

A set function $f: 2^\V \mapsto \RR$ is a function that maps any subset of $\V$ into $\RR$, where $\V$ is the ground set. A submodular function is a subset of set functions with an additional property described as follows.

\begin{definition}[Marginal gain]
	For a set function $f: 2^\V \mapsto \RR$, let $f(e \mid \S) \triangleq f(\S \cup e) - f(\S)$ be the marginal gain of $e \in \V$ with respect to the set $\S$, where $\S \subseteq \V$.
\end{definition}

\begin{definition}[Submodularity \cite{schrijver2003combinatorial}]
	A submodular function $f: 2^\V \mapsto \RR$ is a set function with following property:
	\begin{equation*}
		f(\A \cup v) - f(\A) \ge f(\B \cup v) - f(\B),
	\end{equation*}
	where $\A, \B \subseteq \V, \A \subseteq \B$, and $v\footnote{For notation clarity, we use $e$ to denote $\{e\}$.} \in \V \setminus \B$.
\end{definition}

In addition, a function $f(\cdot)$ is non-decreasing if $f(\A) \le f(\B)$ when $\A \subseteq \B$. Usually, we can use non-decreasing submodular functions to model the utilities in robotics as this marginal gain decreasing property is a very common property in many applications \cite{williams2017decentralized,schrijver2003combinatorial}.

\begin{definition}[Curvature \cite{conforti1984submodular}]
	\label{def: curvature}
	Let $f: 2^\V \mapsto \RR$ be a monotone non-decreasing submodular function, the curvature of $f(\cdot)$ is defined as
	\begin{equation*}
		c_f \triangleq 1 - \min_{a \in \A} \frac{f(\V) - f\left(\V \setminus v\right)}{f(a)},
	\end{equation*}
	where $\A = \{ a \in \V \mid f(a) > 0\}$ and $\V$ is the ground set.
\end{definition}

Curvature measures the degree of the submodularity for a set function, which can be used as a metric when the submodularity of that function is explored \cite{conforti1984submodular}. For any function $f(\cdot)$, it holds that $c_f \in [0, 1]$.

\begin{definition}[\cite{oxley2006matroid}]
	A matroid $\M = (\V, \I)$ is a pair $(V, \I)$, where $\V$ is the ground set, and $\I$ is a collection of subsets of $\V$, with the following properties:
	\begin{enumerate}
		\item $\emptyset \in \I$;
		\item If $\A \subseteq \B \in \I$, then $\A \in \I$;
		\item If $\A, \B \in \I$ with $|\B| < |\A|$, there exists an element $v \in \A \setminus \B$ such that $\B \cup \{x\} \in \I$.
	\end{enumerate}
	\label{def: matroid}
\end{definition}

Matroids generalize the independence property from linear algebra to set systems, which we can use to model the independence relationship between different sets. We direct the reader to \cite{schrijver2003combinatorial} for more theoretical details. Also, examples of using matroids to model independence in robotics can be found in \cite{williams2017decentralized,wehbe2020optimizing,liu2019submodular}.


\subsection{Problem Formulation}
\label{ssec: problem formulation}

\emph{Agents, actions, and constraints:} Consider a  multi-agent system $\A$ with $|\A| = N$, and a common action set $\V$ with $|\V| = M$. Each agent $a \in \A$ can select one or more actions from the action set $\V$. There is a reward associated with an agent's action selection set $\S$, where $\S \subseteq \V$. We denote by $h_i: 2^{\V} \mapsto \RR$ a submodular function as agent $i$'s \emph{individual} objective function for all $i \in \A$. This setting happens when all $i \in \A$ need to cooperate to finish a system objective, and the system redundancy is increased when the system uses multiple agents. We use a matroid $\M = (\V, \I)$ to model the system constraint. For example, such a matroid constraint can be used to model the budget constraint for the system. e.g., $|\S \cap \V| \le z$, where $\S \subseteq \V$ is a problem solution and $z \in \RR$ is the system budget.  Importantly, we make no assumptions about the type of matroid constraint used.

\emph{Attack:} In this multi-agent system, where each agent $i \in \A$ can be deployed to fulfill our system objective, we want to make sure that the system's performance is guaranteed if an attacker attacks any agent by removing its contribution. In such a case, we want to ensure that the minimum performance of any individual agent's performance $h_i(\S)$ is maximized. That is, the minimum performance of the system, $\min_i h_i(\S)$, needs to be maximized. Considering this, we formulate the problem as follows.

\begin{problem}[Robust action selection in multi-agent systems]
Consider a multi-agent agent system $\A$ with $N$ agents selecting actions from a common action set $\V$. Each agent has an individual objective function $h_i: 2^\V \mapsto \RR$ for selecting actions from the action ground set $\V$. The system's constraint is a matroid constraint $\M = (\V, \I)$. The problem is formulated as
\begin{equation*}
	\begin{split}
		\underset{\S \subseteq \V}{\text{maximize}} \quad & \min_i h_i(\S),\\
		\text{subject to}\quad & \S \in \I.
	\end{split}
\end{equation*}
where $i=\{1, \ldots, N\}$, $|\V| =M$, and $\M = (\V, \I)$ is matroid constraint of the multi-agent system.
\label{prb: formulation}
\end{problem}

In this problem formulation, different agents $a \in \A$ need to select a common action set from $\V$ to fulfill a system objective to remain robust. For notation convenience, we denote by
\begin{equation*}
	g(\S)  = \min_i h_i(\S), \quad \forall \S \subseteq \V
\end{equation*}
as the system objective function, and refer to $h_i(\S)$ as the \emph{individual} objective function for agent $i$. The system needs to ensure a performance guarantee if an attacker attacks any agent when knowing the system's common solution $\S$. On the contrary, if we do not require that different agents select a common action set to optimize a system objective, the problem becomes $N$ independent optimization problems. An attacker needs to know more about the system when an attack happens as different agents hold different actions, and the corresponding contributions are different.

The problem formulation describes a game between the system and an attacker that can attack one agent. The system wants to select a common action set $\S$ for all agents to maximize any agent's performance, while an attacker wants to attack one of the agents' contributions in the worst case. Our mission is to ensure the system's performance if the attacker attacks the system in a worst-case scenario.


\section{A Fast Algorithm for Robust Selection}

In general, since the objective function $g(\S) = \min_i h_i(\cdot)$ is a set function and has no obvious property that we can utilize, we resort to the use of a surrogate function with a convenient property that can help service our maximizing purpose. Initially, this surrogate function idea is from \cite{krause2007selecting} when the constraint is relaxed to yield a provable performance bound. However, by exploring the surrogate function using the method proposed in this paper, we can significantly reduce the computational complexity \emph{without} relaxing the system constraint.

The surrogate function $f(\S)$ that can be used for replacing the objective function $g(\S) = \min_i h_i(\cdot)$ is as follows:
\begin{equation*}
	f(\S) \leftarrow \frac{1}{N} \sum_{i=1}^N \min \left\{h_i(\S), \gamma\right\},
\end{equation*}
where $N$ is the number of agents or the number of objective functions, and $\gamma$ is used as the upper bound for the associated reward of different agents. It can be proven that the surrogate function $f(\S)$ is submodular \cite{fujito2000approximation}.

The proposed fast algorithm for solving Problem \ref{prb: formulation} is shown in \cref{alg1}. First, we set an upper and lower bound for our surrogate function $f(\cdot)$. The lower bound $\ell$ is set to $0$, while the upper bound is set to $f(\V)$ as shown in \cref{lin: 1}. The upper bound equals the function value when all actions are selected, which is the maximum performance that the system can obtain.

Since the surrogate function $f(\cdot)$ is a monotone non-decreasing function, we can use a binary search method to reduce the gap between the upper bound $u$ and the lower bound $\ell$. First, we set the median of our surrogate function $f(\cdot)$ as $\gamma \leftarrow u +\ell$ as shown in \cref{lin: median}. Then, we use this median $\gamma$ as a temporary upper bound of the individual objective function $h_i(\S)$ for different agents. Meanwhile, the corresponding surrogate function $f(\cdot)$ needs to be updated as $f(\S) \leftarrow \frac{1}{N} \sum_{i=1}^N \min \left\{h_i(\S), \gamma\right\}$. Based on the upper bound $\gamma$, we will then need to find an action set $\S$ that approximately maximizes the surrogate function $f(\cdot)$ as maximizing a submodular function $f(\cdot)$ with a matroid constraint is NP-hard \cite{schrijver2003combinatorial}.

In \cref{lin: for 1} to \cref{lin: for 2}, we depict a fast method for maximizing a submodular function under a matroid constraint. Conventionally, maximizing a submodular function is achieved through a simple greedy method \cite{fisher1978analysis},
\begin{equation*}
	\S \leftarrow \S \cup \left\{ \argmax_{\S \cup \{e\} \in \I} f(e \mid \S) \right\}, \quad \forall e \in \V \setminus \S.
\end{equation*}
This method aims to build a solution set $\S$ by adding the element $e \in \V \setminus \S$ with the maximum marginal gain in each iteration. The computational complexity of this method is $\O(M^2)$ when the size of the action ground set is $|\V| = M$. However, when the size of the ground set is large, or it is time-consuming to evaluate the function $f(\cdot)$, we need a faster way to generate the final solution set. In general, our method can be summarized as follows:
\begin{equation*}
	\S \leftarrow \S \cup \left\{e \mid f(e \mid \S) \ge \Delta, \S \cup e \in \I\right\}, \quad \forall e \in \V \setminus \S.
\end{equation*}
The parameter $\Delta$ is a lower bound for deciding whether to add $e \in \V \setminus \S$ to the solution $\S$ or not based on the current marginal gain $f(e \mid \S)$. Note that the basic problem constraint $\S \cup \{e\} \in \I$ should also be followed in each iteration.

Specifically, from \cref{lin: for 1} to \cref{lin: for 2}, we first set a lower bound for the marginal gain $f(e \mid \S)$ based on the current system solution $\S$. If both this lower bound and the system matroid constraint $\M = (\V, \I)$ are satisfied, we can add this $e$ to the current solution set $\S$ as shown in \cref{lin: for update}. After all available $e \in \V \setminus \S$ are added to $\S$, we then update the marginal gain lower bound $\Delta$ as
\begin{equation*}
	\Delta \leftarrow \frac{\Delta}{1+\delta}.
\end{equation*}
This iterative process terminates when all $e \in \V$ are selected or when the lower bound of the marginal gain $f(e \mid \S)$ achieves its predefined minimum value $\delta F$. Using this method, we do not need to calculate the maximum value in each iteration. Instead, we can add all elements that satisfy our requirements into the current solution set $\S$, making the algorithm faster. When compared with the conventional greedy maximization method, which has a complexity of $\O(M^2)$, the method (from \cref{lin: for 1} to \cref{lin: for 2}) has a worst case complexity of $\O(M \log(\delta^{-1}))$.

Finally, we need to check the surrogate function value $f(\S)$ against our predefined lower bound $\frac{\gamma}{1 + c_f + \delta}$, where $c_f$ is the curvature of $f(\cdot)$. Then, we use a binary search method to update those two bounds. If the function value $f(\S)$ is lower than the predefined bound, we update the upper bound using $\gamma$. On the contrary, if the objective function value $f(\S)$ is larger than the predefined bound $\frac{\gamma}{1 + c_f + \delta}$, we need to update the lower bound $\ell$ using $\gamma$, and store the current solution $\S$ as the best solution $\S^{\text{G}}$ as shown in \cref{lin: update S}.

\begin{algorithm}[!t]
	\caption{A Fast Algorithm for Robust Selection}
	\label{alg1}
	\textbf{Input:} The inputs are as follows:
	\begin{itemize}
		\item The individual objective function $h_i(\cdot)$;
		\item The action ground set $\V$.
	\end{itemize}

	\textbf{Output:} Set $\S^{\text{G}}$.

	\begin{algorithmic}[1]
		\Statex
		\State $\ell \leftarrow 0, u \leftarrow \min_i h_i(\V)$; \label{lin: 1}
		\While{$|u-\ell| > \epsilon$}
		\State $\S \leftarrow \emptyset$;
		\State $\gamma \leftarrow \frac{1}{2}(u + \ell)$; \label{lin: median}
		\State $F \leftarrow \max_{i \in \V} f(i)$;
		\State
		\For{$\Delta = F, e \in \V \setminus \S; \Delta \ge \delta F; \Delta \leftarrow \frac{\Delta}{1+\delta}$} \label{lin: for 1}
		\State $\S \leftarrow \S \cup \{e \mid f(e \mid \S) \ge \Delta, \S \cup e \in \I\}$;\label{lin: for update}
		\EndFor \label{lin: for 2}
		\State
		\If{$f(\S) < \frac{\gamma}{1 + c_f + \delta} $} \label{lin: update 1}
		\State $u \leftarrow \gamma$;
		\Else
		\State  $\ell \leftarrow \gamma, \S^{\text{G}} \leftarrow \S$; \label{lin: update S}
		\EndIf \label{lin: update 2}
		\EndWhile
		\State \textbf{return} $\S^{\text{G}}$.
	\end{algorithmic}
\end{algorithm}

\begin{theorem}
	\label{thm: relation}
	For any fixed $\gamma$, let $\S^\star$ be an optimal solution for maximizing the surrogate function $f(\cdot)$, and $\S_m$ be the solution in the $m$th iteration in the for loop as shown in \cref{lin: for update}, then we have the following
	\begin{equation*}
		f(e \mid \S_m) \ge \frac{1}{1 + \delta} f(o \mid \S_m), \quad \forall o \in \S^\star \setminus \S_m, e \in \V \setminus \S_m.
	\end{equation*}
\end{theorem}

\begin{proof}
	Denote by $\Delta_n$ the lower bound of the marginal gain in the $n$th for loop as shown in \cref{lin: for update}. Note that it is possible that $n \neq m$ as for any $\Delta_n$ it is possible that $\S$ is updated multiple times. We also denote by $\S_{m-1}$ the solution before the updating using the current lower bound $\Delta_n$. For any $e \in \V \setminus \S_{m-1}$ that can be added to the approximation solution as $\S_m \leftarrow \S_{m-1} \cup e$, there are two cases to be considered: (1). $e$ is the first selected element that satisfies the above two requirements; (2). $e$ is the element other than the first one that satisfies the above two requirements.

	In the first case, for any $e \in \V \setminus \S_{m-1}$ that can be added to $\S_{m-1}$, it should satisfy the following condition
	\begin{equation}
		f(e \mid \S_{m-1}) \ge \Delta_n, \quad \forall e \in \V \setminus \S_{m-1}.
		\label{eq: 11}
	\end{equation}
	This is the necessary condition for $e \in \V \setminus \S_{m-1}$ to be considered as a candidate for adding to $\S_{m-1}$, as shown in \cref{lin: for update}. Now, for any $o \in \S^\star \setminus \S_{m-1}$, we have
	\begin{equation*}
		f(o \mid \S_{m-1}) \le \Delta_{n-1}, \quad \forall o \in \S^\star \setminus \S_{m-1}.
	\end{equation*}
	That is because if $o \in \S_{m-1}$, we have $f(o \mid \S_{m-1}) = 0$. Then, the above statement is true. If $o \notin \S_{m-1}$, it means $f(o \mid \S_{m-1})$ cannot satisfy the lower bound condition $\Delta_{n-1}$ in the $(j-1)$th iteration. Also, we have the relationship between lower bounds in different iterations as $\Delta_n \leftarrow \frac{\Delta_{n-1}}{1+\delta}$. We then have
	\begin{equation}
		f(o \mid \S_{m-1}) \le \Delta_{n-1} =  (1 + \delta) \Delta_n, \quad o \in \S^\star \setminus \S_{m-1}.
		\label{eq: 12}
	\end{equation}
	Combining \cref{eq: 11} and \cref{eq: 12}, for any $o \in \S^\star \setminus \S_{m-1}$, we have the result for the first case.
	\begin{equation*}
		f(e \mid \S_{m-1}) \ge \frac{1}{1 + \delta} f(o \mid \S_{m-1}), \quad \forall e \in \V \setminus \S_{m-1}.
	\end{equation*}

	In the second case, we consider the case where $e$ is not the first element that satisfies the current marginal gain lower bound $\Delta_n$. Denote by $\S_m$ the solution after the first eligible element $e$ added to $\S$. Without loss of generality, we assume $e \in \V \setminus \S_m$. Following the same reasoning above, for any $o \in \S^\star \setminus \S_m$, we then have
	\begin{equation}
		f(e \mid \S_m) \ge \frac{1}{1 + \delta} f(o \mid \S_{m-1}), \quad \forall e \in \V \setminus \S_m.
		\label{eq: 21}
	\end{equation}
	At the same time, by utilizing the submodularity of the objective function $f(\cdot)$, we have
	\begin{equation}
		f(o \mid \S_{m-1}) \ge f(o \mid \S_m).
		\label{eq: 22}
	\end{equation}
	Combining \cref{eq: 21} and \cref{eq: 22}, for any $o \in \S^\star \setminus \S_m$, we have the following result for the second case.
	\begin{equation*}
		f(e \mid \S_m) \ge \frac{1}{1 + \delta} f(o \mid \S_m), \quad \forall e \in \V \setminus \S_m.
	\end{equation*}

	Finally, we can complete the proof by reduction based on the above two cases.
\end{proof}

\begin{theorem}
	Let $c_f$ be the curvature of function $f(\cdot)$. For any fixed $\gamma$, let $\S^\star$ be an optimal solution for maximizing the surrogate function $f(\cdot)$ and $\S$ be the generated output using the approximation method (\cref{lin: for 1} to \cref{lin: for 2}), we then have the following.
	\begin{equation*}
		f(\S) \ge \frac{1}{1 + c_f + \delta} f(\S^\star).
	\end{equation*}
	\label{thm: 2}
\end{theorem}

\begin{proof}
	First, by utilizing the submodularity of function $f(\cdot)$, we have
	\begin{equation*}
		f(\S^\star \cup \S) \le f(\S) + \sum_{o \in \S^\star \setminus \S} f(o \mid \S).
	\end{equation*}
	Also, we assume that $|\S^\star| = K$. Then, we have
	\begin{equation*}
		f(\S^\star \cup \S) \le f(\S) + \sum_{i=1}^K f(o \mid \S_{m-1}), \forall o \in \S^\star \setminus \S_{m-1}.
	\end{equation*}
	By using the result from \cref{thm: relation}, we have
	\begin{equation}
	\label{eq: alg 1}
		\begin{split}
			& f(\S^\star \cup \S) \\
			\le & f(\S) + (1+\delta) \sum_{i=1}^K f(e \mid \S_{m-1}), \forall e \in \V \setminus \S_{m-1} \\
			\le & f(\S) + (1+\delta) \sum_{i=1}^K f(\S_m \mid \S_{m-1}).
		\end{split}
	\end{equation}

	At the same time, for the element $e \in \S_m \setminus \S_{m-1}$, we have
	\begin{equation*}
	\begin{split}
		& f(\S^\star \cup \S) \\
		= &  f(\S^\star) + \sum_{i=1}^K \left[f(\S^\star \cup \S_m) - f(\S^\star \cup \S_{m-1})\right] \\
		= &  f(\S^\star) + \sum_{i=1}^K \left[f(\S^\star \cup \S_{m-1} \cup e) - f(\S^\star \cup \S_{m-1})\right].
	\end{split}
	\end{equation*}
	By utilizing the definition of curvature for submodular functions, as shown in \cref{def: curvature}, we then have
	\begin{equation*}
		\begin{split}
			& f(\S^\star \cup \S_{m-1} \cup e) - f(\S^\star \cup \S_{m-1}) \\
			\ge & f(\S_{m-1} \cup e) - f(\S^\star \cup \S_{m-1}).
		\end{split}
	\end{equation*}
	By combining the above two results, we get
	\begin{equation}
	\label{eq: alg 2}
		f(\S^\star \cup \S) \ge \sum_{i=1}^K \left[f(\S_{m-1} \cup e) - f(\cup \S_{m-1})\right]
	\end{equation}

	Finally, combining \cref{eq: alg 1} and \cref{eq: alg 2}, we get the final result.
	\begin{equation*}
		f(\S) \ge \frac{1}{1 + c_f + \delta} f(\S^\star).
	\end{equation*}
\end{proof}

\begin{theorem}
    In \cref{alg1}, let $\S^{\text{G}}$ be the solution and $\S^\star$ be an corresponding optimal solution for the same $\gamma$. When the algorithm terminates, we then have the following.
    \begin{equation*}
        \min_i h_i(\S^{\text{G}}) \ge \frac{1}{1 + c_f + \delta} \min_i h_i(\S^\star).
    \end{equation*}
\end{theorem}

\begin{proof}
	From \cref{thm: 2}, for any fixed $\gamma$, we know that $f(\S) \ge \frac{1}{1 + c_f + \delta} f(\S^\star)$. Also, since $f(\S) \leftarrow \frac{1}{N} \sum_{i=1}^N \min \left\{h_i(\S), \gamma\right\}$, we then have that both $h_i(\S)$ and $\gamma$ should be larger than $f(\S^\star)$ according to the definition of $f(\cdot)$. We then have the above result.
\end{proof}


\section{Evaluation}

To evaluate the performance of the proposed algorithm, we use an action section problem as a study case. Specifically, we use a sensor proximity maximization problem \cite{hou2021robust}. The problem settings are as follows. The agents and the sensors are located in $\RR^2$ with a size of $100 \times 100$. The locations of both the agents and the actions are randomly chosen. The idea of sensor proximity is to use the sensing radius instead of the covered area to calculate the effectiveness of the selected sensor. Based on the locations of the agents, we need to select the locations of sensors to maximize our objective function. In this context, we refer to sensors as actions, and agents need to select a set of actions from the action ground set. This problem can be viewed as a specific example of an action selection problem.

We denote by $\A$ and $\V$ the agent set and the action ground set. The objective function can be formulated as
\begin{equation*}
	h_i(\S) = \max_{j \in \S} d_{ij},
\end{equation*}
where $i \in \A$ is the $i$th agent, and $j \in \S$ is the $j$th action. The 2D distance function $d: \RR^2 \times \RR^2 \mapsto \RR$ is a Euclidean distance function in $\RR^2$. Given any selected action set $\S \subseteq \V$, the objective value of $i$th agent is defined as the maximum distance between the agent $i$ and all the actions in the set $\S$. Meanwhile, this $100 \times 100$ region is equally divided into $4$ sub-regions. There is an upper limit on the number of actions that can be selected by each agent in each sub-region. This upper limit serves as a budget for the system in each sub-region and can be formulated as a matroid constraint $\M = (\V, \I)$. Thus, the problem formulation is
\begin{equation*}
	\begin{split}
		\underset{\S \subseteq \V}{\text{maximize}} \quad & \min_i h_i(\S),\\
		\text{subject to}\quad & \S \in \I.
	\end{split}
\end{equation*}
where $h_i(\S) = \max_{j \in \S} d_{ij}$ is the individual objective function for agent $i$. Specifically, the matroid constraint $\M = (\V, \I)$ of this problem is a partition matroid \cite{schrijver2003combinatorial} and can be written as
\begin{equation*}
	\I: |\S \cap \V_j| \le z, \quad j = 1, \ldots, 4,
\end{equation*}
where $\V_j$ is the sub-ground set of the problem associated with the $j$th sub-region, and $\bigcup_{j=1}^4 \V_j = \V$. Note that the total number of sub-regions is $4$.

\begin{figure}[!tbp]
	\centering
	\includegraphics[width=3in]{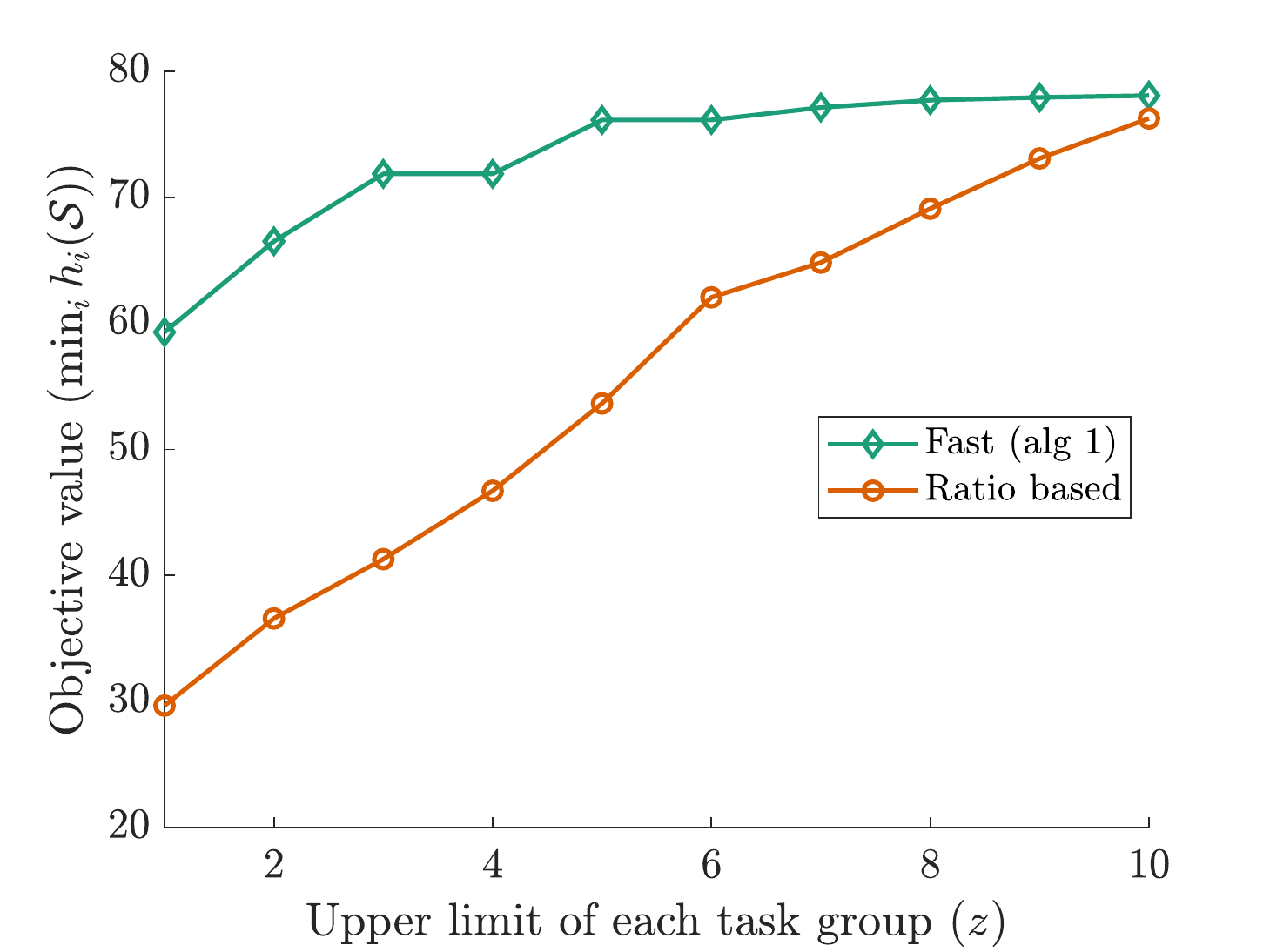}
	\caption{The averaged objective value comparison between the proposed algorithm: ``Fast (alg 1)'' and the ratio based algorithm. Each averaged objective value is calculated based on $100$ trials.}
	\label{fig: value}
\end{figure}

Next, we conduct Monte Carlo simulations to test the performance. Specifically, the number of agents is set to $|\A| = 5$ and the number of actions is set to $|\V| = 50$. The locations of agents and actions are randomly located. Therefore, the number of actions in each sub-region is not necessarily equal. That is, it is likely that $|\V_j| \neq |\V_{j'}|$ if $j \neq j'$ and $j, j' = 1, \ldots, 4$. In the simulation, we change the upper limit $z$ from $1$ to $10$. For each settled upper limit, we then run the simulation $100$ times, where the locations of both the agents and the actions are randomly generated for each simulation. The curvature of $f(\cdot)$ is set to $1$, which is the upper bound. The tunable parameter $\delta$ is set to $10^{-3}$. To test the performance, we compare the performance of the proposed \cref{alg1} with the current state-of-the-art algorithm from \cite{hou2021robust}. In that algorithm, the final solution is generated based on the ratios of the contributions of different actions with respect to the maximum contribution of each agent's available selection. We refer to this method as the ``Ratio based'' method in the comparison. Meanwhile, we refer to the algorithm proposed in this paper as ``Fast (alg 1)''. Then, we compare the performance and the speed of the two different methods.

Since we run the simulation $100$ trials for each random setting from $z = 1$ to $z = 10$, we calculate the averaged objective value to test the performance. From different $z$'s, we get the averaged objective function value comparison, which is shown in \cref{fig: value}. The result shows that the proposed algorithm performs much better when the limit $z$ is set to lower values. Then, when the limit $z$ approaches $10$, we get almost identical performance.

Next, we compare the speed of the proposed method with the ratio-based method. Specifically, we compare the averaged number of objective function evaluations. For each upper limit setting $z$, we calculate the number of evaluations of $f(\cdot)$. Then, we have the speed comparison as shown in \cref{fig: num}. This result indicates that the proposed fast method is significantly faster than the ratio-based method, even when $z$ is set to a low value. Then, when we increase the upper limit $z$, we observe that the difference between the two speeds becomes more apparent. Meanwhile, when compared with the proposed method, we notice that speed of the ratio-based method does not change too much when $z$ is different. This is because we always need to find the action with the maximum contribution from the action ground set $\V$ in the corresponding method. This result shows the speed superiority of the proposed method.

\begin{figure}[!tbp]
	\centering
	\includegraphics[width=3in]{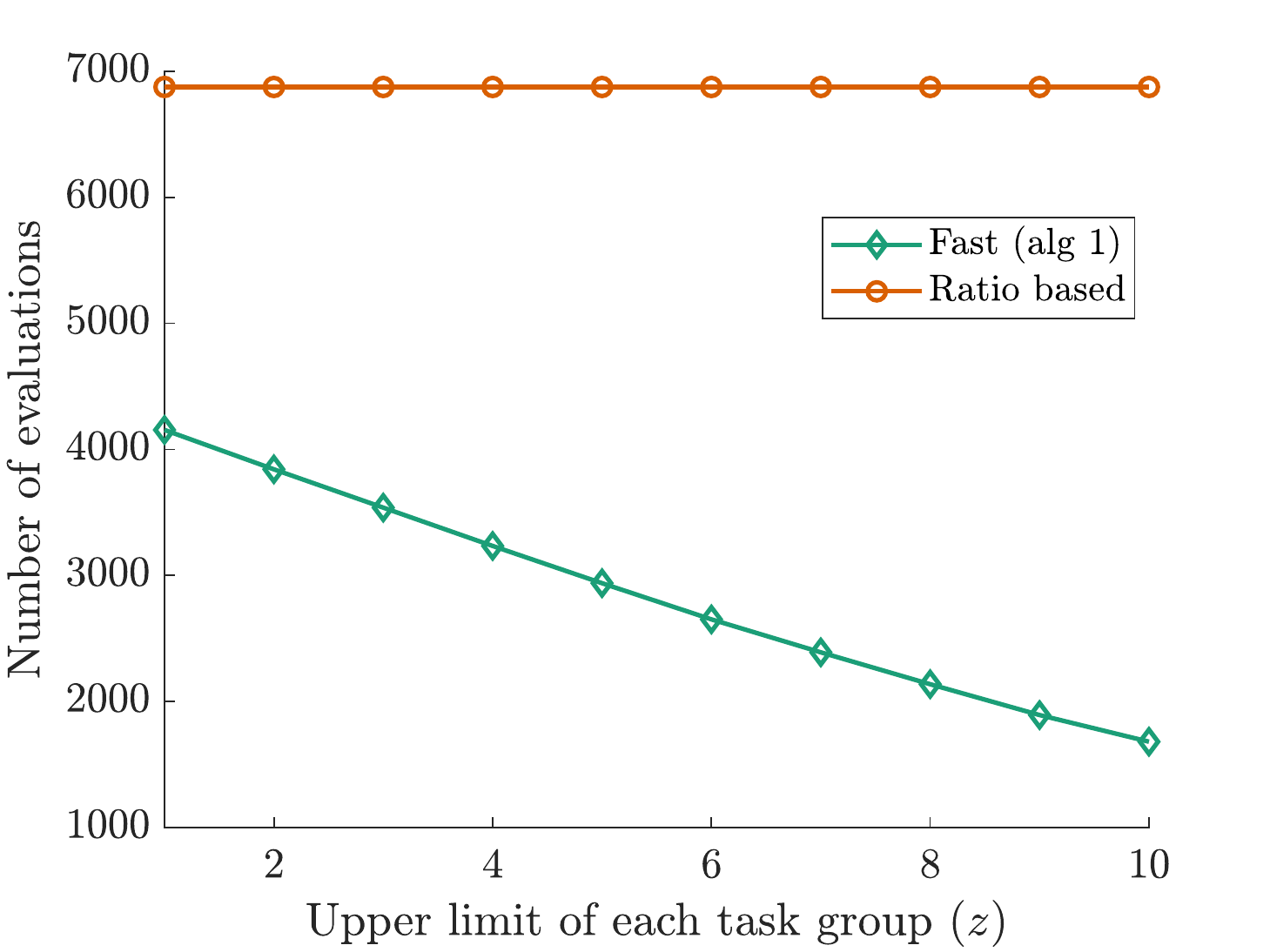}
	\caption{The number of objective function evaluations comparison between the proposed algorithm: ``Fast (alg 1)'' and the ratio based algorithm. Each averaged objective value is calculated based on $100$ trials.}
	\label{fig: num}
\end{figure}


\section{Conclusions}

This paper proposed a fast algorithm for solving the robust action selection problem for multi-agent systems. Through the proposed method, we can significantly increase the algorithm's speed. The proposed algorithm can be used when it is time-consuming to calculate the objective function value or the size of the ground set is large. Finally, through Monte Carlo simulation, we are able to evaluate the performance of the proposed algorithm through comparisons with a state-of-the-art algorithm.


\bibliography{arxiv.bib}
\bibliographystyle{IEEEtran}

\end{document}